\documentclass{cccg17}
\usepackage{graphicx,amssymb,amsmath}
\usepackage{url,float}
\usepackage{amsfonts}
\usepackage{latexsym}
\usepackage[boxruled]{algorithm2e}
\usepackage{tcolorbox} 

\usepackage{color}
\newcommand{\bluenew}[1]{{#1}}  
\newcommand{\cyannew}[1]{{#1}}  
\newcommand{\rednote}[1]{{}}  

\usepackage[normalem]{ulem}

\newcommand{\lemlab}[1]{\label{lemma:#1}}
\newcommand{\thmlab}[1]{\label{thm:#1}}
\newcommand{\figlab}[1]{\label{fig:#1}}
\newcommand{\seclab}[1]{\label{sec:#1}}

\newcommand{\lemref}[1]{\ref{lemma:#1}}
\newcommand{\thmref}[1]{\ref{thm:#1}}
\newcommand{\secref}[1]{\ref{sec:#1}}
\newcommand{\figref}[1]{\ref{fig:#1}}

\def\F{{\mathcal F}}

\def\D{{\Delta}}
\def\e{{\varepsilon}}
\def\a{{\alpha}}
\def\b{{\beta}}
\def\q{{\theta}}

\def\bG{{\partial G}}

\newcommand{\squeezelist}{\setlength{\itemsep}{0pt}}

\title{Angle-monotone Paths\\
in Non-obtuse Triangulations}

\author{
Anna Lubiw%
    \thanks{School of Computer Science, University of Waterloo, Waterloo, Ontario, Canada.
      \protect\url{alubiw@uwaterloo.ca}.}
\and 
Joseph O'Rourke%
    \thanks{Department of Computer Science, Smith College, Northampton, MA, USA.
      \protect\url{jorourke@smith.edu}.}
}

\index{Lubiw, Anna}
\index{O'Rourke, Joseph}

\begin{document}
\maketitle

\begin{abstract}
We reprove a 
result of Dehkordi, Frati, and Gudmundsson:
every two vertices in a non-obtuse triangulation of a point set
are connected by an angle-monotone path---an $xy$-monotone path
in an appropriately rotated coordinate system. We show that this result cannot be
extended to angle-monotone spanning trees, but can be extended to
boundary-rooted spanning forests. The latter leads to
a conjectural edge-unfolding
of sufficiently shallow polyhedral convex caps.
\end{abstract}

\section{Introduction}
\seclab{Introduction}
The central result of this paper is to offer an alternative---and we believe simpler---proof
of a result of
Dehkordi, Frati, and Gudmundsson~\cite{dfg-icgps-15} (henceforth, DFG): 
\begin{tcolorbox}
``\textbf{Lemma~4}. Let $G$ be a Gabriel triangulation on a point set $P$.
For every two points $s,t \in P$, 
there exists an angle $\q$ such that $G$ contains a $\q$-path from $s$ to $t$.''
\end{tcolorbox}
We first explain this result, using notation from~\cite{bbcklv-gtamg-16},
before detailing our other contributions.
First, we use $S$ for the point set and $\b$ instead of $\q$.
A Gabriel triangulation as defined by DFG is a triangulation of $S$ where
``every angle of a triangle delimiting an internal face is acute.''
Because neither they nor we need any of the various properties of Gabriel triangulations
except the angle property, and we only need non-obtuse rather than strict acuteness,
we define $G$ to be a plane geometric graph that is a non-obtuse triangulation of $S$.

Define the \emph{wedge} $W(\b,v)$ to be the region of
the plane 
bounded by rays at angles $\b \pm 45^\circ$ 
emanating from $v$.
$W$ is closed along (i.e., includes) both rays, and has 
angular  \emph{width} of $90^\circ$.
(Later we generalize to widths different from $90^\circ$.)
A polygonal path $(v_0,\ldots,v_k)$ \bluenew{consisting of} edges of $G$ is called \emph{$\b$-monotone}
(for short, a \emph{$\b$-path})
if the vector of every edge $(v_i, v_{i+1})$ lies in $W(\b,v_0)$.
These are the $\q$-paths of DFG.
Note that if $\b=45^\circ$, then a $\b$-monotone path is
both $x$- and $y$-monotone with respect to a Cartesian coordinate system.
A path that is $\b$-monotone for some $\b$ is called \emph{angle-monotone}.

Our phrasing of the DFG result is: 

\begin{theorem}
In a non-obtuse triangulation $G$, 
every pair of vertices is connected by an angle-monotone path.
\thmlab{PairPath}
\end{theorem}

\noindent{\bf Other Contributions.}  We extend Theorem~\thmref{PairPath} to wedges of any width $\gamma$---if a plane geometric graph that includes the convex hull of $S$ has all angles at most $\gamma$, 
then there is an angle-monotone path of width $\gamma$ between any two vertices.
\cyannew{Of necessity, $60^\circ \le \gamma < 180^\circ$.}
One significance of angle-monotone paths of width $\gamma < 180^\circ$ is that they have a spanning ratio of $1/\cos{\frac{\gamma}{2}}$~\cite{bbcklv-gtamg-16}.  We do not pursue that aspect here.  
Instead, we investigate angle-monotone spanning trees.
\bluenew{These were studied independently in~\cite{mastakas2016rooted}, which addressed
recognition and construction, but not existence---our focus.}
We show that Theorem~\thmref{PairPath} does not extend to angle-monotone spanning trees, but  does extend to boundary-rooted spanning forests.  
Then, in Section~\secref{Unfolding} we make a novel connection to edge-unfolding polyhedra.

\section{Proof}
\seclab{Proof}
We prove Theorem~\thmref{PairPath} by showing that there is an angle-monotone path from an arbitrary fixed vertex $s$ to every other vertex. 
The proof uses an angular sweep \bluenew{around} $s$,
which by convention we place  
at the origin. 
\cyannew{We first consider 
a fixed \bluenew{but arbitrary} angle $\b$ and
investigate} 
\rednote{AL: just more active verbs.}
which vertices are reached by $\b$-paths from $s$. 
Let $\bG$ be the boundary of $G$, i.e., the convex hull of $S$.  
Our \bluenew{proof relies} on two properties of vertices $v$ not on $\bG$: 
(1) the wedge $W(\b,v)$ includes at least one edge incident to $v$; 
(2) if the wedge has only one edge incident to $v$ then that edge does not lie along a bounding ray of the wedge.
\rednote{Rev1:
``I found a bit confusing that properties (1) and (2) are stated after an angle $\beta$ has been fixed. Now it is true that $\beta$ has no specific value; yet I guess it would be better to state properties (1) and (2) for *any* angle $\beta$, and then to fix $\beta$ to any specific value.''
JOR: Minor confusion. I fixed with the ``\bluenew{but arbitrary}" above.}
In order to avoid dealing with boundary vertices as a separate case, we will augment $G$ so that conditions~(1) and~(2) hold for boundary vertices as well.  
At every vertex $v$ on $\bG$ add 
a finite set of rays that subdivide the exterior angle at $v$ into angles of at most $90^\circ$.
Call the result $G^+$.  By construction, properties (1) and (2) now hold for every vertex $v$ if we consider both edges and rays incident to $v$. 
Note that no added ray crosses an edge of $G$.  In the special case of 
\bluenew{widths $\ge 90^\circ$}, the rays can be chosen so that they do not intersect 
one another.
\rednote{Rev1: ``I guess this is true for any angle $\ge 90^\circ$, and not only for $90^\circ$."
JOR: Good point, we just use $=90^\circ$ even for $\gamma > 90^\circ$. Do not think it worthwhile
to explain, so just a silent change.}
When we generalize to smaller widths, the rays will necessarily intersect each other but this will not  
influence our proof.  In our figures, the rays are drawn short as a reminder that only their angles at the convex hull are relevant.

A $\b$-path starting at $s$ is \emph{maximal} if there is no edge or ray that can be added at the end of the path while keeping it a $\b$-path.  In particular, a path ending with a ray is maximal.  DFG proved that every maximal $\b$-path terminates on $\bG$, which in our terms becomes:

\begin{lemma}
Any maximal $\b$-path ends with a ray.
\lemlab{maximal-beta}
\end{lemma}
\begin{proof}
Consider a $\b$-path that ends at a vertex $v$.  The wedge $W(\b, v)$ must include an edge or ray of $G^+$ by property~(1) so the path can be extended.   
\end{proof}

\medskip
\noindent
We will see later (in \bluenew{Fig.~\figref{Beta1Beta2}(b))}
\rednote{JOR: Only illustrated in (b).}
that it is possible for a $\b$-path to 
include edges of $\bG$, 
\rednote{Rev1: and that those boundary edges may even form multiple subpaths. ``Not clear what you mean by this."
JOR: Agree, too terse. Attempted fix:}
\bluenew{return to interior edges, and then later again include edges of $\bG$}.
For a fixed $\b$, let $P(\b)$ be the set of all maximal $\b$-paths starting from $s$.
Let $V(\b)$ and $E(\b)$ be the set of vertices and edges/rays in $P(\b)$.

Let $U(\b) \in P(\b)$ be the 
\emph{upper envelope} of $P(\b)$,
defined as follows.
Starting from vertex $v=s$,  grow \bluenew{$U(\b)$} 
by always following
the most counterclockwise edge/ray from $v$ falling within $W(\b,v)$.
$U(\b)$ is necessarily a maximal $\b$-path, so it ends with a ray.  By property~(2) above, we have:

\begin{obs}
$U(\b)$ does not include any edge/ray along the lower ray of \bluenew{$W(\b,s)$}, at angle $\b - 45^\circ$.
\label{obs:no-lower}
\end{obs}

$L(\b)$ is similarly the \emph{lower envelope}, the most clockwise path.
Note that ``upper'' and ``lower'' are to be interpreted 
as most counterclockwise and most clockwise respectively, not in terms of $y$-coordinates.%
\footnote{These notions are analogous but not equivalent to DFG's ``high'' and ``low" paths.}
Finally, define $R(\b)$ to be the region of the plane whose boundary is $L(\b)$, $U(\b)$.
Fig.~\figref{PbetaExample} illustrates these definitions.
\begin{figure}[htbp]
\centering
\includegraphics[width=0.75\linewidth]{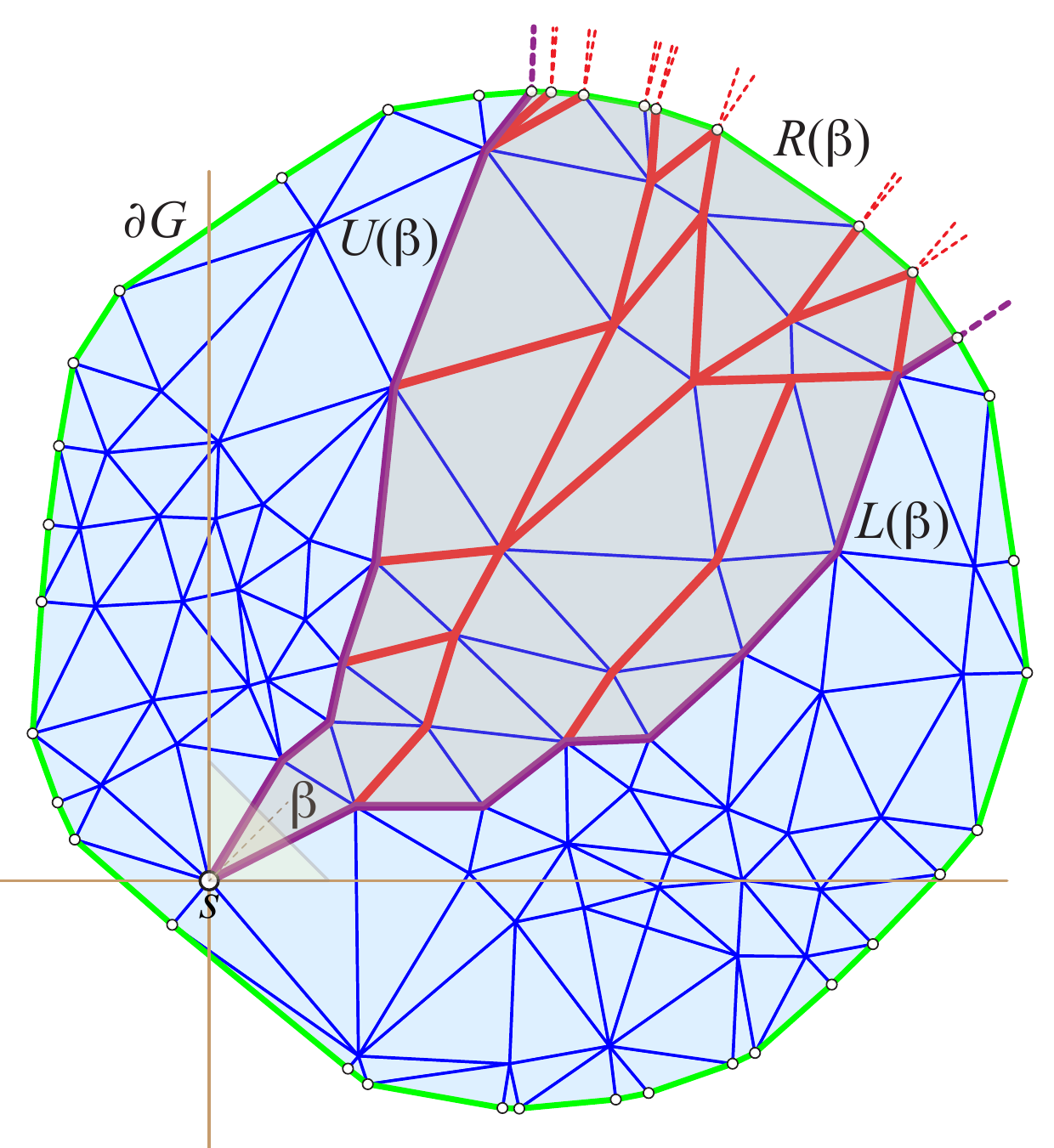}
\caption{$P(\b)$ edges are marked; $\b=45^\circ$.
$L(\b)$ and $U(\b)$ delimit the region $R(\b)$.
Rays shown only for $V(\b)$ hull vertices.
}
\figlab{PbetaExample}
\end{figure}

\begin{lemma}
Every vertex in $R(\b)$ is in $V(\b)$, i.e., every
vertex in $R(\b)$ can be reached from $s$ via a $\b$-path.
\lemlab{RincludesVerts}
\end{lemma}
\begin{proof}
Wlog assume $\b=45^\circ$, so that the wedge rays are at $0$ and $90^\circ$.
Let $v \in V(\b)$ be the leftmost \emph{inaccessible} vertex, i.e., the leftmost vertex not
reached by a $\b$-path from $s$. 
If there are ties for leftmost, let $v$ be the lowest.
Consider the \bluenew{backward} wedge $\overline{W}(\b,v)$ at $v$. 
Note that $s$ lies in $\overline{W}(\b,v)$.  Consider the line segment $sv$.  It lies in $\overline{W}(\b,v)$ and inside the convex hull of $S$.  Imagine rotating $sv$ clockwise or counterclockwise about $v$ while remaining inside $\overline{W}(\b,v)$ and inside the convex hull of $S$ in a small neighborhood of $v$. 
Since there are no obtuse angles at $v$, rotating in one direction or the other must result in an edge or ray in  $\overline{W}(\b,v)$.
Furthermore,  the result cannot be a ray since we never leave the convex hull.  Thus we have identified an edge $e=(u,v)$ in $\overline{W}(\b,v)$.
\rednote{Rev1: ``I guess this part is unnecessarily complicated. The existence of this edge comes from property (1) applied to wedge $W(\pi+\beta,v)$."
JOR: I think I prefer the text as-is, bringing in the hull, and not-a-ray.
And I don't particularly like using $W(\pi+\beta,v)$.
But I will let you decide on whether to revise this argument.  AL: Yes, I agree, let's leave our text as is.}

Suppose first that $u$ is in $R(\b)$. Because $v$ is the leftmost lowest inaccessible
vertex, $u$ must be accessible 
(note that $u$ must be at \bluenew{the same $y$-height} or lower than $v$).
But now $v$ lies in $W(\b,u)$, and so \bluenew{$v$} is accessible after all, a contradiction.
Instead suppose $u$ lies outside $R(\b)$.
Then $e$ must cross the boundary of $R(\b)$. 
But that boundary is composed of edges/rays of $G^+$, and $e$ cannot cross
an edge of $G^+$ without the two sharing a vertex, which would lie on the
boundary of $R(\b)$, not the exterior, again a contradiction.
\end{proof}

\begin{figure}[htbp]
\centering
\includegraphics[width=0.75\linewidth]{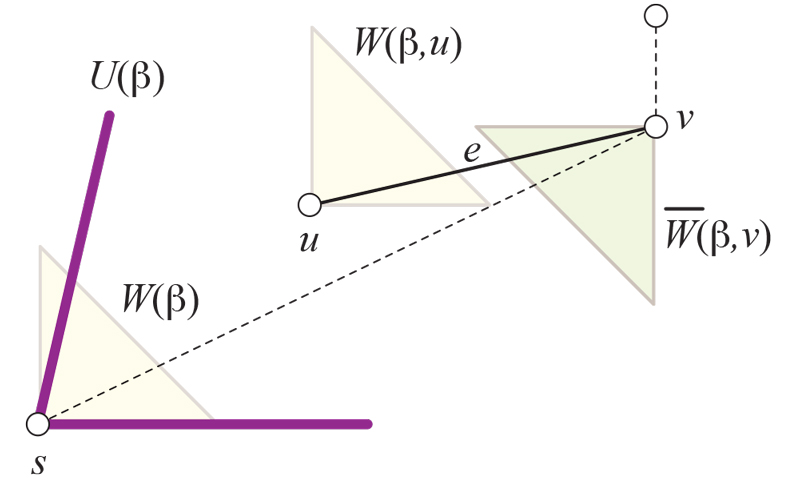}
\caption{No vertex in $R(\b)$ is inaccessible: all are reached by a $\b$-path from $s$.
}
\figlab{NotInaccessible}
\end{figure}

\rednote{Rev1: ``The first paragraph of the "Critical angles $\beta_i$" section is quite critical! I would rewrite it carefully. Some specific comments follow."}
\subsection{Critical angles $\b_i$}
\seclab{Critical}
We now 
analyze the relationships between $P(\b_i)$ and $P(\b_{i+1})$,
where $\b_{i+1}$ is the next ``relevant'' angle after $\b_i$,
with the goal of showing that all vertices in $G$ are
\bluenew{``covered'' over all $\b_i$,
i.e., belong to $P(\b_i)$ for some $\b_i$.}
\bluenew{Throughout, fix the source $s$, and let $W(\b) = W(\b,s)$.}
\rednote{JOR: Because reviewer wants $s$ or $v$ explicit.}
Define an angle $\b$ to be \emph{critical} if $P(\b+\e)$ or $P(\b-\e)$ differs
from $P(\b)$, for an arbitrarily small $\e > 0$.
\bluenew{
At a critical angle $\b$, one or both rays
of $W(\b)$ are parallel to one or more edges of $P(\b)$ 
}
\rednote{Rev1: ``It is not clear here what rays you are talking about. In fact, for these assertions to be true, you need to consider all the wedges centered at the vertices of $P(\beta)$ and not just the one centered at $s$. In these two sentences I would talk about angles; something like: "one or more edges with slope $\beta-45^\circ$ drop out of $P(\beta)$".
JOR: I don't agree with the reviewer. The reviewer is emphasizing $W(\b,v)$ as fundamentally
different from $W(\b) = W(\b,s)$, but that's not true for directions.
I think we just have to avoid implying that an edge lies directly on a bounding ray of the wedge.
So use ``parallel to" instead of ``lies along."
AL: Thank you, good fixes.
}
If $P(\b+\e)$ differs from $P(\b)$, one or more edges \bluenew{parallel to}
the $\b-45^\circ$ ray drop out of $P(\b)$.
If $P(\b-\e)$ differs from $P(\b)$, one or more edges \bluenew{parallel to}
the $\b+45^\circ$ ray enter $P(\b)$.
\bluenew{Fig.~\figref{Beta1Beta2} illustrates two adjacent critical angles.}

Let $\b_1, \b_2, \ldots, \b_i, \b_{i+1}, \ldots$ be the critical angles,
sorted counterclockwise.
For every $\b$ strictly between two adjacent critical angles,
$\b_i < \b < \b_{i+1}$, $P(\b)$ is the same.
We use the notation $P(\b_{i'})$ to represent this
intermediate set, which differs from $P(\b_i)$ or $P(\b_{i+1})$ or both.

In the transition from $P(\b)$ to $P(\b+\e)$, edges can drop out of $P(\b)$.
In particular, any edge $e=(u,v)$ that lies along the $\b-45^\circ$-ray of 
\bluenew{$W(\b,u)$} will drop out of $P(\b)$.  Furthermore, when edges drop out of $P(\b)$ this may cause vertices to drop out of $P(\b)$, and any edge originating at a dropped vertex also drops out of $P(\b)$. 
\begin{figure}[htbp]
\hspace*{-4mm}%
\centering
\includegraphics[width=1.1\linewidth]{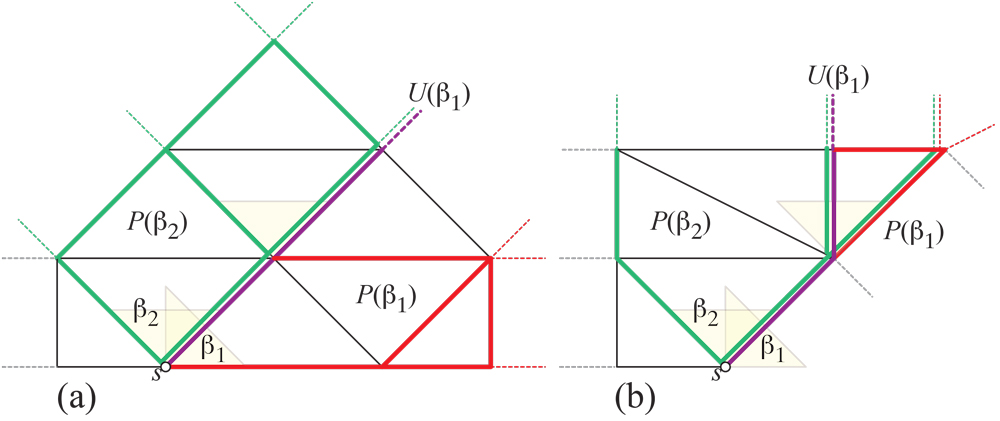}
\caption{$P(\b_1)$ (red), $U(\b_1)$ (purple), and $P(\b_2)$ (green).
(a)~$P(\b_1)$ includes some but not all $\bG$ edges.
(b)~$U(\b_1)$ includes an edge of $\bG$ before an internal edge.}
\figlab{Beta1Beta2}
\end{figure}

The next lemma shows that no vertices
fall strictly ``between'' $P(\b_i)$ and $P(\b_{i+1})$, 
where they would escape being spanned.

\begin{lemma}
$U(\b_i)$ is a path in $P(\b_{i+1})$.
\lemlab{UbetaShared}
\end{lemma}

\begin{proof}
Since edges may only enter, not leave, in the transition from $P(\b_{i+1}-\e)$ to $P(\b_{i+1})$,
the lemma is equivalent to the claim that  $U(\b_i)$ 
is a path in $P(\b_{i}+\e) = P(\b_{i'})$.
Let $\b=\b_i$, and $\b' = \b_{i'} $.
Then we aim to prove that
$U(\b) \subseteq P(\b')$.
This requires showing \bluenew{that} no edge $e \in U(\b)$ drops out from $P(\b)$ to $P(\b')$,
as $\b$ increases to $\b'$.
As usual, assume that $\b=45^\circ$.

Suppose to the contrary that some edge drops out,
and let $e=(u,v)$ be the leftmost lowest \bluenew{edge} with $e \in U(\b)$ but $e \notin P(\b')$.
Equivalently, $e$ is the first edge of $U(\b)$ that is not in $P(\b')$.
Because $u$ is in $P(\b')$, the only reason for $e$ to drop out is that it
lies along the lower, horizontal ray of the wedge $W(\b,u)$.
But by Observation~(\ref{obs:no-lower}), $U(\b)$ does not include any edge along the lower ray of the wedge.
\end{proof}

In analogy with the definition of $R(\b)$,
define $R(\b_i, \b_j)$ for $j > i$ to be the region bound by 
$L(\b_i)$, $U(\b_j)$, and the portion
of $\bG$ between those lower and upper envelope endpoints.

\begin{lemma}
$R(\b_i, \b_{i+1}) = R(\b_i) \cup R(\b_{i+1})$.
Informally, no vertices are ``orphaned'' between $P(\b_i)$ and $P(\b_{i+1})$.
\lemlab{NoOrphans}
\end{lemma}
\begin{proof}
\rednote{AL: We could avoid the fuss by taking the "Informally . . " sentence out of the Lemma statement and giving as proof:}
\cyannew{The lemma essentially says that no vertices are ``orphaned'' between $P(\b_i)$ and $P(\b_{i+1})$, and this follows immediately from the fact that $U(\b_i) \subseteq P(\b_{i+1})$ as established in Lemma~\lemref{UbetaShared}.}
%
\end{proof}

\noindent
Now we can prove Theorem~\thmref{PairPath}, the key result of~\cite{dfg-icgps-15}:
\begin{proof}[of Theorem~\thmref{PairPath}]
Fix $s$ and construct $\bigcup_i P(\b_i)$.
By Lemma~\lemref{NoOrphans}, this is a spanning graph of $G$, and so must
include a path from $s$ to $v$.
\end{proof}

Our arguments extend to wedges of any width $\gamma$, thus proving that 
if a plane geometric graph that includes the convex hull of $S$ has all internal angles at most $\gamma$, then there is an angle-monotone path of width $\gamma$ between any two vertices.
This answers a question raised in~\cite{bbcklv-gtamg-16}.

\rednote{Rev1:
``Section 3 suffers from the fact that you forgot to define angle-monotone spanning trees. I take that you mean the following: For a fixed node s, the required spanning tree T(s) contains an angle-monotone path between s and every other node. But the name leaves room for several variations (is the angle of monotonicity required to be the same for all the paths in the tree? Are the paths required to exist between any pair of nodes and not only from a source s to all the other nodes?), so please introduce a definition."}
\section{Spanning Tree}
\seclab{SpanningTree}
Now that Theorem~\thmref{PairPath} has established that there is
a graph spanning all of $G$ with angle-monotone paths from any source $s$,
it is natural to wonder if the claim can be strengthened to the existence of
\cyannew{an \emph{angle-monotone spanning tree}} for any $s$: 
\bluenew{a tree rooted at $s$ with an angle-monotone path
from $s$ to any $v \in G$.}
The answer is {\sc no}, but we canvass a few positive results
before detailing a counterexample for spanning trees.
\bluenew{Throughout, we let $s$ be an arbitrary vertex of $G$.}
\rednote{JOR: Because reviewer seems to need this explicit.}
First, within a fixed $\b$ region, $P(\b)$ can be easily spanned:
\begin{lemma}
$P(\b)$ includes \bluenew{an angle-monotone tree} that spans the same vertices, $V(\b)$.
\lemlab{PbetaStree}
\end{lemma}
\begin{proof}
\bluenew{By Lemma~\lemref{RincludesVerts}, $P(\b)$ reaches every vertex in $V(\b)$.}
For each vertex $v \in V(\b)$, in any order, delete all but one incoming edge to $v$.
Because an incoming edge remains to each $v$, $v$ is spanned.
Because eventually no $v$ has more than one incoming edge, no cycles can remain.
See Fig.~\figref{PbetaStree}. 
\end{proof}
\begin{figure}[htbp]
\centering
\includegraphics[width=0.5\linewidth]{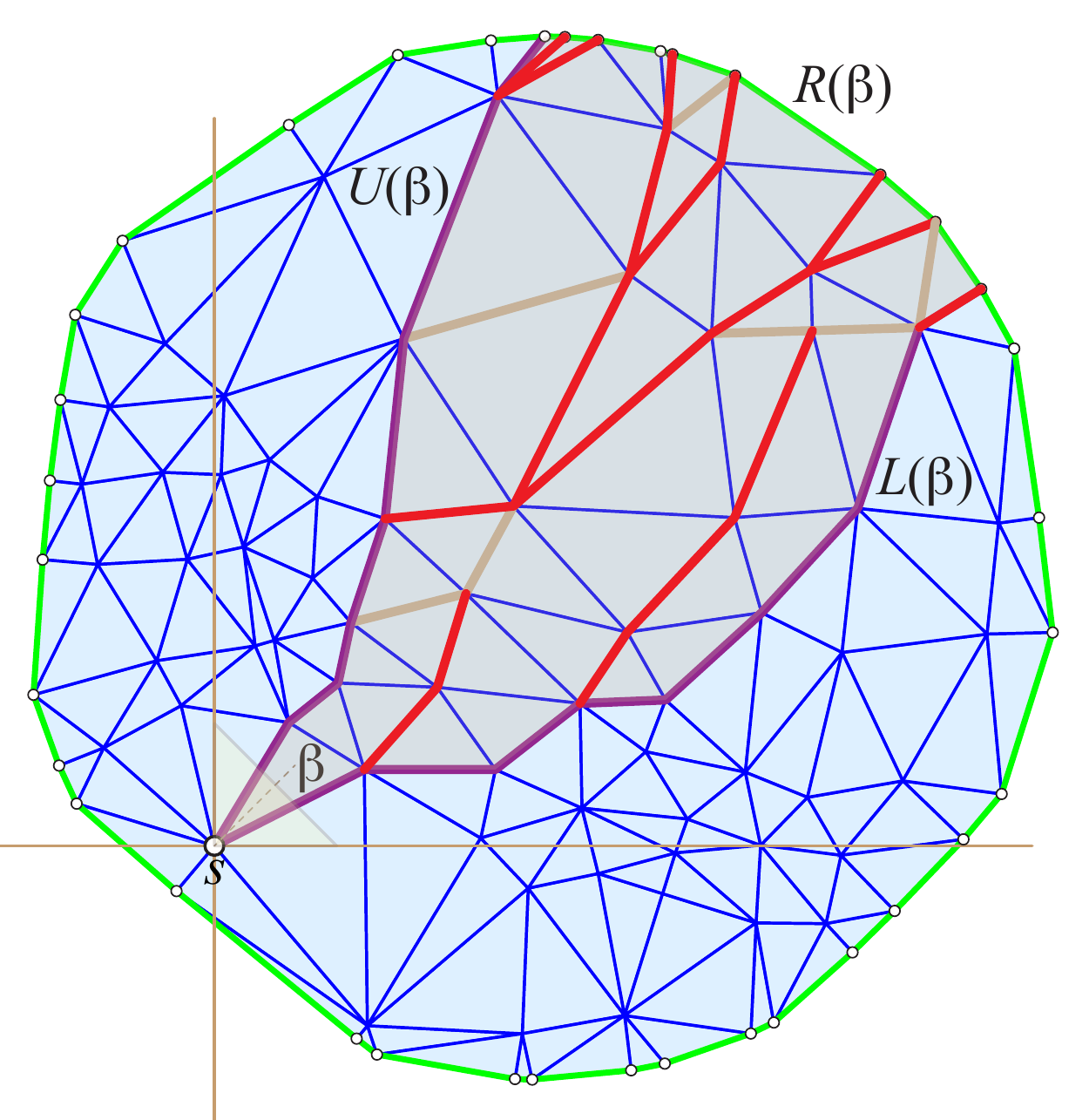}
\caption{$P(\b)$ spanning tree. Light-brown edges have been deleted.}
\figlab{PbetaStree}
\end{figure}

We now consider triangulations with special angles.

\begin{lemma}
Let $G_{45^\circ}$ have edges only at multiples of $45^\circ$.
Then there is an angle-monotone
spanning tree rooted at any source vertex $s$.
\lemlab{k45}
\end{lemma}
\begin{proof}
Let $\b_1$ and $\b_2 = \b_1 + 45^\circ$ be two consecutive critical angles.
We argue that $U(\b_1)$ and $U(\b_2)$ may share an initial portion of a path,
but then diverge and do not rejoin before reaching 
their terminal rays.

By Observation~(\ref{obs:no-lower}), $U(\b)$
only includes edges at angles $\b_i$ or $\b_i+45^\circ$.
So the most counterclockwise edge in 
$U(\b_1)$
is $\b_1+45^\circ$,
and the most clockwise edge in 
$U(\b_2)$
is $\b_2 = \b_1 + 45^\circ$.
Thus, $U(\b_1)$ and $U(\b_2)$
can \bluenew{have parallel edges}, but once
they separate, they can never rejoin.

Now it is easy to create a spanning tree between 
$U(\b_1)$ and $U(\b_2)$
that retains all edges in these two envelopes, by deleting all but one
incoming edge to each vertex between the envelopes.
\end{proof}

\medskip
\noindent
\bluenew{The problematic possibility avoided in such $G_{45^\circ}$ graphs
is $U(\b_1)$ and $U(\b_2)$ joining, separating, and rejoining.}
Already in a graph $G_{30^\circ}$ that has edges only at multiples of $30^\circ$,
the divergence of upper envelopes used in Lemma~\lemref{k45} is no longer
guaranteed, and thwarts that proof.

\rednote{JOR:
I flipped Fig.~\figref{Ubeta_k45} horizontally to save space, but if
we cannot fit into $6$ pages, this figure could go.}

\subsection{Spanning Tree Counterexample}
\seclab{Cex}
Fig.~\figref{SpanningTreeCex} shows a graph $G$ that does not have
an angle-monotone spanning tree rooted at $s$.
The construction allows two \cyannew{angle-monotone} paths to each of $\{C,D,E,F\}$, one of which is
marked green in the figure. But vertices $A$ and $B$ are shifted slightly
toward one another, which breaks the symmetry and, \cyannew{as we shall argue below}, results in a unique angle-monotone path to each. 
\bluenew{The union of those two unique paths contains the cycle $(s,a,x,b)$.}
\cyannew{Thus there is no angle-monotone spanning tree from $s$.}

\cyannew{We now argue that there is no angle-monotone path to $A$ other than $saxA$.
This is simply a matter of checking that any other path to $A$ contains two \emph{spread-apart} edges whose vectors do not lie in a  $90^\circ$ wedge.  In particular, the path $sbxA$ contains spread-apart edges $sb$ and $xA$, and 
the path $sawA$ contains spread-apart edges $aw$ and $wA$.  Other paths can be checked similarly.}
\rednote{AL: Maybe the above argument is clearer.  If you agree, we can delete the next paragraph.}
Similar reasoning constrains the \bluenew{(symmetric)} paths to $B$.
\begin{figure}[htbp]
\centering
\includegraphics[width=1.05\linewidth]{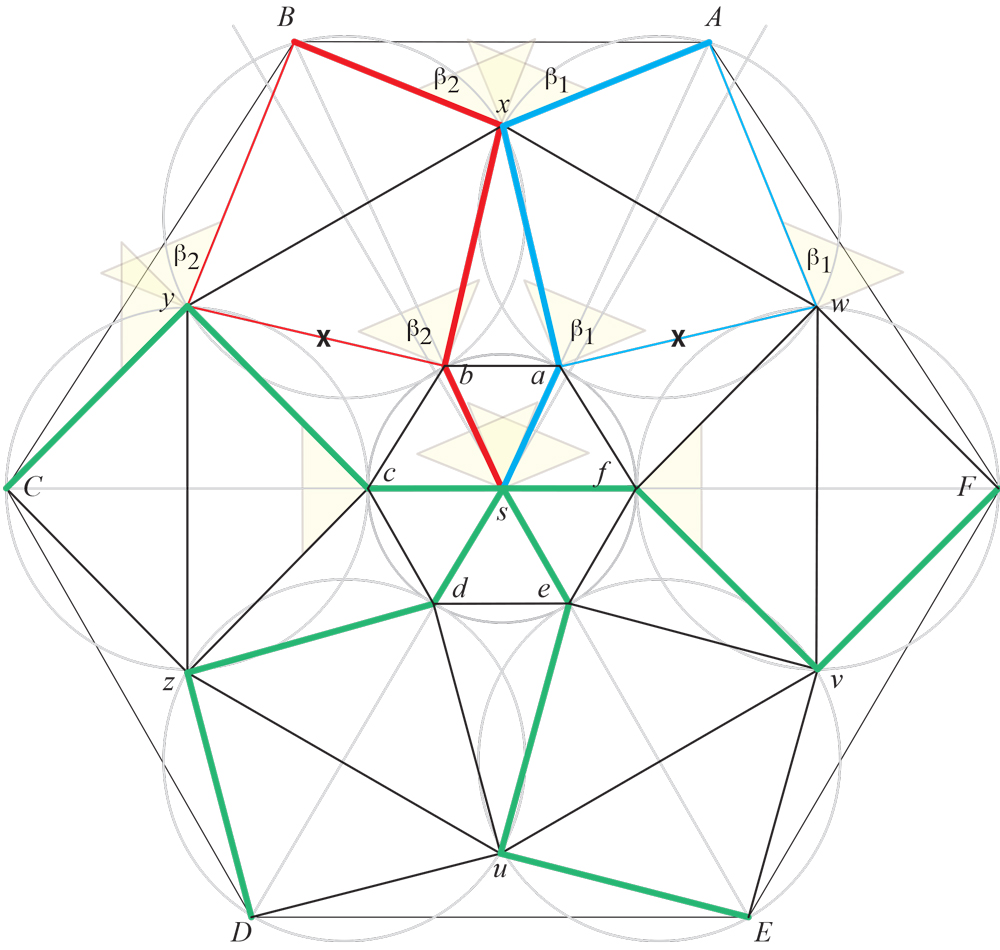}
\caption{$(s,a,x,A)$ is the unique \cyannew{angle-monotone} 
path to $A$,
and $(s,b,x,B)$ is the unique \bluenew{angle-monotone} path to $B$, \bluenew{forming the cycle $(s,a,x,b)$}.}
\figlab{SpanningTreeCex}
\end{figure}
\rednote{Rev1:
``Given the contents of Section 2, it would be much nicer to have a counter-example to the existence of a spanning tree in a Gabriel triangulation (the one you present looks like a Gabriel triangulation, but there are some disturbing angles equal to $90^\circ$ and the outer boundary is not convex)."
JOR: I added edges to make it convex, and I am pretty certain they do not break
the example, but I left in the $90^\circ$ angles.
I believe that the reviewer is asking is: How robust is our cex? Does it require $=90^\circ$ angles?
Good question. I think it *is* robust, but this needs thought... Here is my conclusion:}

\bluenew{The outer ring of six circles in the construction make clear that various
diameter-spanning angles are $90^\circ$, but points
$\{a,A, \ldots, f,F\}$ could be moved slightly exterior to those circles, rendering those angles
$< 90^\circ$, while retaining the properties that force the $(s,a,x,b)$ cycle.
So the counterexample is ``robust" in this sense.}
\rednote{AL: Thank you, this seems good.}

\section{Spanning Forest}
\seclab{SpanningForest}
For the unfolding application discussed in the next section, it is useful to
span $G$ by a boundary-rooted forest $\F$: A set of disjoint \cyannew{angle-monotone} 
trees,
each with its root on $\bG$, and spanning every interior vertex of $G$.
This can be achieved with \cyannew{$\b$-monotone trees for} just four $\b$ values.

With a Cartesian coordinate system centered on vertex $s$, define the
quadrants $Q_0,Q_1,Q_2,Q_3$ as follows.
$Q_0$ is the quadrant coincident with
$W(\b,s)$ when $\b=45^\circ$, closed along the $x$-axis
and open along the $y$-axis, and includes the origin $s$. 
$Q_i$, $i>0$ are defined analogously,
except those quadrants do not include the origin.
Thus the quadrants are pairwise disjoint and together cover the plane.

We construct separate spanning forests for each quadrant,
following Algorithm~1, which grows paths from vertices interior to $G$
to $\bG$. 
See Fig.~\figref{SpanningForest}.
\begin{figure}[htbp]
\centering
\includegraphics[width=0.75\linewidth]{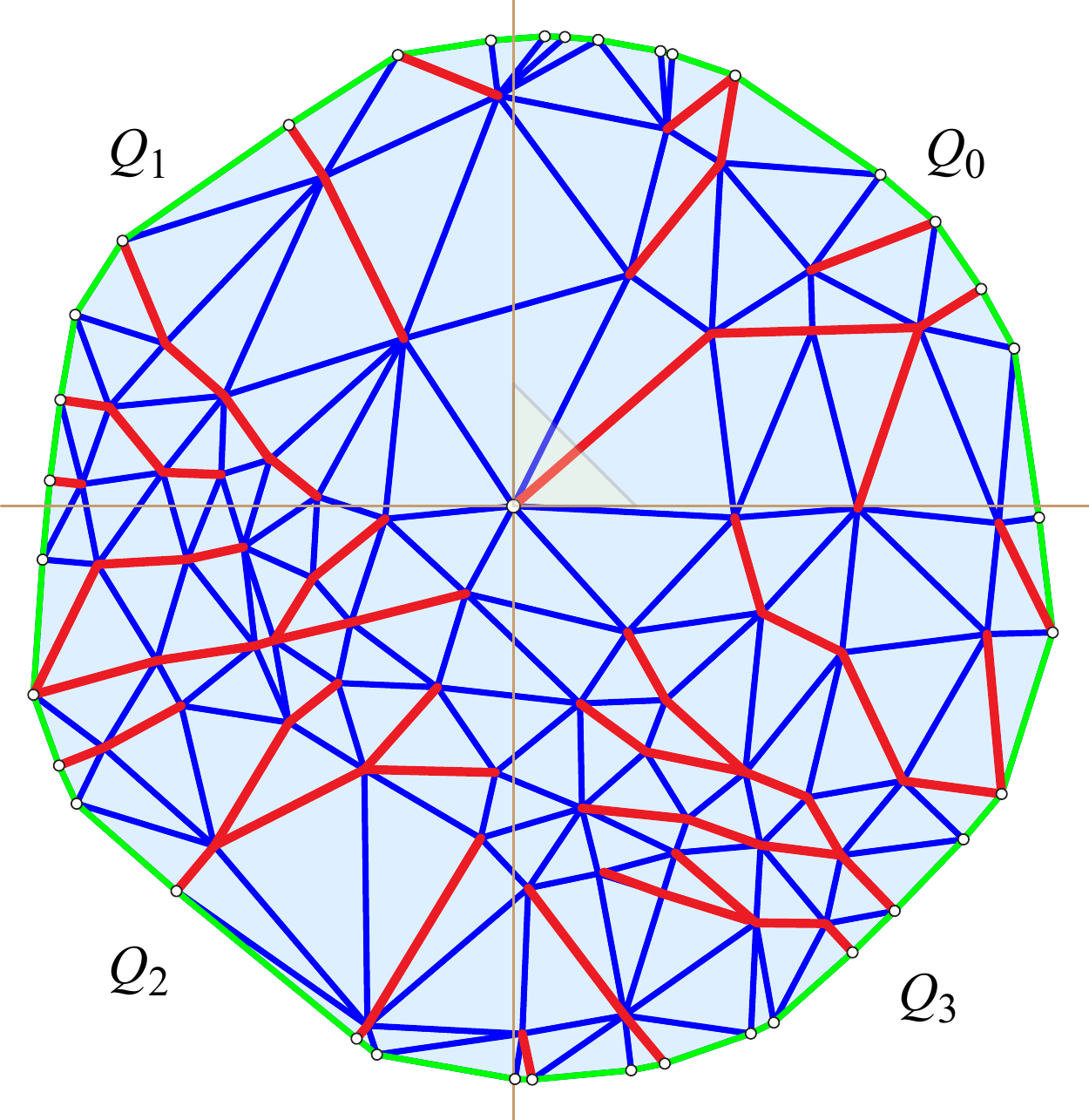}
\caption{Spanning forest resulting from Algorithm~1.}
\figlab{SpanningForest}
\end{figure}

\begin{algorithm}[htbp]
\caption{Algorithm to construct spanning forest $\F$}
\DontPrintSemicolon 
    \SetKwInOut{Input}{Input}
    \SetKwInOut{Output}{Output}

    \Input{Non-obtuse triangulation $G$}
    \Output{Spanning forest $\F$ composed of $\b$-monotone paths}
    
    \BlankLine
    \tcp{Quadrants $Q_j$, each \bluenew{corresponding to}}
    \tcp{$\b_j = 45^\circ + j \cdot 90^\circ$, $j=0,1,2,3$.}
    \BlankLine
    
     \ForEach{Quadrant $Q_j$, $j=0,1,2,3$}{
     
         $F_j \leftarrow \varnothing$
         
         \tcp{Grow forest $F_j$ inside $Q_j$.}
    	 \ForEach{$v \in Q_j$}{
	 
	 \If{$v \notin F_j$}
	 {Grow $\b_j$-path $p$ from $v$.\;
	 Stop when $p$ reaches a vertex in $F_j$, or reaches $\bG$.}
	 
	 $F_j \leftarrow F_j \cup p$
   	  
   	  }

     }
     \BlankLine
     
     $\F = F_0 \cup F_1 \cup F_2 \cup F_3$\;
     \KwRet $\F$.
\end{algorithm}

\rednote{Rev1:
``For $j=0,1,2,3$, the trees in $Q_j$ are monotone with respect to the angle $225^\circ + j \cdot 90^\circ$ and not with respect to the angle $45^\circ + j \cdot 90^\circ$ (if I got the definition of $\beta$-monotone trees right and the paths go from the source to the nodes).
JOR: I think this is splitting hairs. The paths are at $45^\circ + j \cdot 90^\circ$, but,
yes, the trees are aimed the other way.}

\rednote{Rev1:
``In the description of Algorithm 1 output and in the statement of Lemma 8, one gets the false impression that the forest is composed of connected components which are paths."
JOR: Certainly the figure removes that confusion, but I added clarification.}

\begin{lemma}
Algorithm~1 outputs a boundary-anchored spanning forest,
\bluenew{each tree of which is}
composed of $\b$-monotone paths, for four $\b$'s:
$\b_j = 45^\circ + j \cdot 90^\circ$, $j=0,1,2,3$.
\lemlab{SpanningForest}
\end{lemma}
\begin{proof}
Observe that a $\b_j$-path grown from $v \in Q_j$ remains in $Q_j$.
So all the trees in $F_j$ are composed of $\b_j$-paths.
All vertices of each quadrant are spanned, because the inner loop
of Algorithm~1 runs over all $v \in Q_j$.
No cycles can be created because the algorithm only grows
a path from $v$ if $v$ is not yet in $F_j$. So $v$ becomes a leaf
of some tree in $F_j$ when its path reaches that tree.
\end{proof}

\section{Unfolding}
\seclab{Unfolding}
Now we \bluenew{discuss} an application of Algorithm~1 and Lemma~\lemref{SpanningForest}
to edge-unfolding nearly flat convex caps.
We \bluenew{only} sketch the argument, 
\bluenew{as several steps need considerable elaboration, and other steps rely}
on definitions and 
details in an unpublished report~\cite{o-ucprm-16}.
So this section will end with a conjecture rather than a theorem.
At a high level, the construction depends on two claims:
(1)~angle-monotone paths are ``radially monotone paths,'' a concept introduced
in~\cite{o-ucprm-16}, but known before as backwards ``self-approaching curves"~\cite{ikl-sac-99}.
(2)~Theorem~2 
of~\cite{o-ucprm-16} concludes that the unfolding
of a particular ``medial" cut path $M$ on a polyhedron is radially monotone and so does not self-cross
when unfolded
(if certain angle conditions are satisfied).

Let $P$ be a convex polyhedron, and let $\phi(f)$ for a face $f$ be
the angle the normal to $f$ makes with the $z$-axis.
Let $H$ be a halfspace whose bounding plane is orthogonal to the $z$-axis, and includes points
vertically above that plane.
Define a \emph{convex cap} $C$ of angle $\Phi$ to be $C=P \cap H$
for some $P$ and $H$, such that $\phi(f) \le \Phi$ for all $f$ in $C$.
We will only consider $\Phi < 90^\circ$, which implies that the projection
$C_\bot$ of $C$ onto the $xy$-plane is one-to-one.

Say that a convex cap $C$ is \emph{acutely triangulated}
if every angle of every face is strictly acute.
Note that $P$ being acutely triangulated does not always imply that 
$C=P \cap H$ is acutely triangulated,
\rednote{Rev1:
``By cutting a simplicial polyhedron with a plane, you might get some faces which are quadrilaterals and not triangles. In there your angles have to be $\ge 90^\circ$. Is there anything that needs to be said about that?"
JOR: No, because the assumption is that $P$ is acutely triangulated, not that some procedure
gives it to us acutely triangulated. (And obtuse angles could be easily acutely triangulated.)
Oh, OK, will cite Bishop.}
\bluenew{
but it is known that any polyhedron can be acutely triangulated~\cite{bishop2016nonobtuse}.}
We will need this lemma.
\begin{lemma}
Let a triangle in $\mathbb{R}^3$, whose face normal makes angle $\phi$
with the \bluenew{$z$-axis}, have one angle $\a$, which projects to $\a_\bot$ on the $xy$-plane.
Then the maximum value of $\D = | \a - \a_\bot |$
is a monotonically increasing function, as plotted in Fig.~\figref{AngleProjectionGraph}.
\lemlab{AngleProjection}
\end{lemma}
We only need that $\D \to 0$ as $\phi \to 0$, so we will not calculate the function explicitly.
For example, for $\phi < 10^\circ$, $\D < 1^\circ$. 
\begin{figure}[htbp]
\centering
\includegraphics[width=0.5\linewidth]{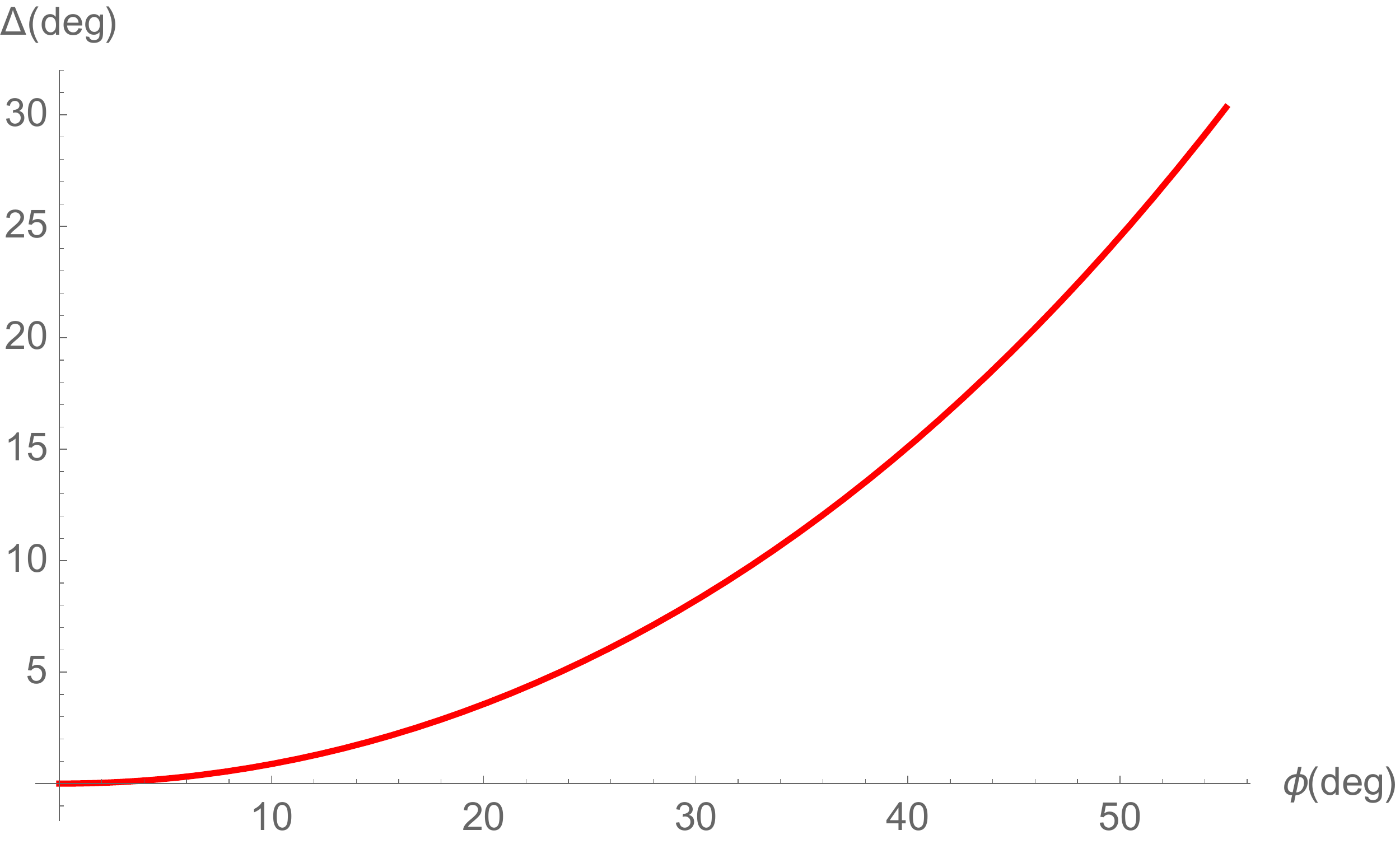}
\caption{The maximum face angle change resulting from projection with normal
at angle $\phi$.}
\figlab{AngleProjectionGraph}
\end{figure}

For a triangulated convex cap $C$, let $\a_\textrm{max}$ be the maximum of any triangle angle.
Using Lemma~\lemref{AngleProjection}, we can guarantee that an acutely triangulated
cap $C$ will project to a non-obtuse plane graph $C_\bot$ 
by choosing $\Phi$ so that $\D < 90^\circ - \a_\textrm{max}$.

Now we apply Algorithm~1 and Lemma~\lemref{SpanningForest} to obtain an angular-monotone
spanning forest $\F_\bot$ of $C_\bot$.
We then lift the trees in $\F_\bot$ to cut trees $\F$ on $C$ in $\mathbb{R}^3$.
Again Lemma~\lemref{AngleProjection} ensures this can be accomplished without
any turn angle in any path in any tree of $\F$ exceeding $90^\circ$.
Now finally we invoke \bluenew{a version of} Theorem~2 as mentioned previously, which guarantees
that the cut paths unfold without local overlap.
We leave it a claim that the angle conditions for that theorem are satisfied 
by selecting $\Phi$ small enough.
The conclusion is that the lifted paths are ``radially monotone,'' which is the
condition that implies unfolding without overlap.
The end result is this:

\begin{conj}
For an acutely triangulated convex cap $C$ with sufficiently small $\Phi$ bounding face normals,
the spanning forest $\F_\bot$ resulting from Algorithm~1
lifts to a cut forest $\F$ that edge-unfolds $C$ without overlap.
\end{conj}

\medskip
\noindent
We have implemented this construction.
Fig.~\figref{AM-20-30-s1-v100-3D} shows a convex cap
with $\Phi \approx 27^\circ$, and Fig.~\figref{AM-20-30-s1-v100-Lay} shows
the corresponding unfolding.\footnote{
The forest in Fig.~\figref{AM-20-30-s1-v100-3D}
is slightly different than that shown in Fig.~\figref{SpanningForest},
due to different ordering choices of $v \in Q_j$.
}
\begin{figure}[htbp]
\centering
\includegraphics[width=0.80\linewidth]{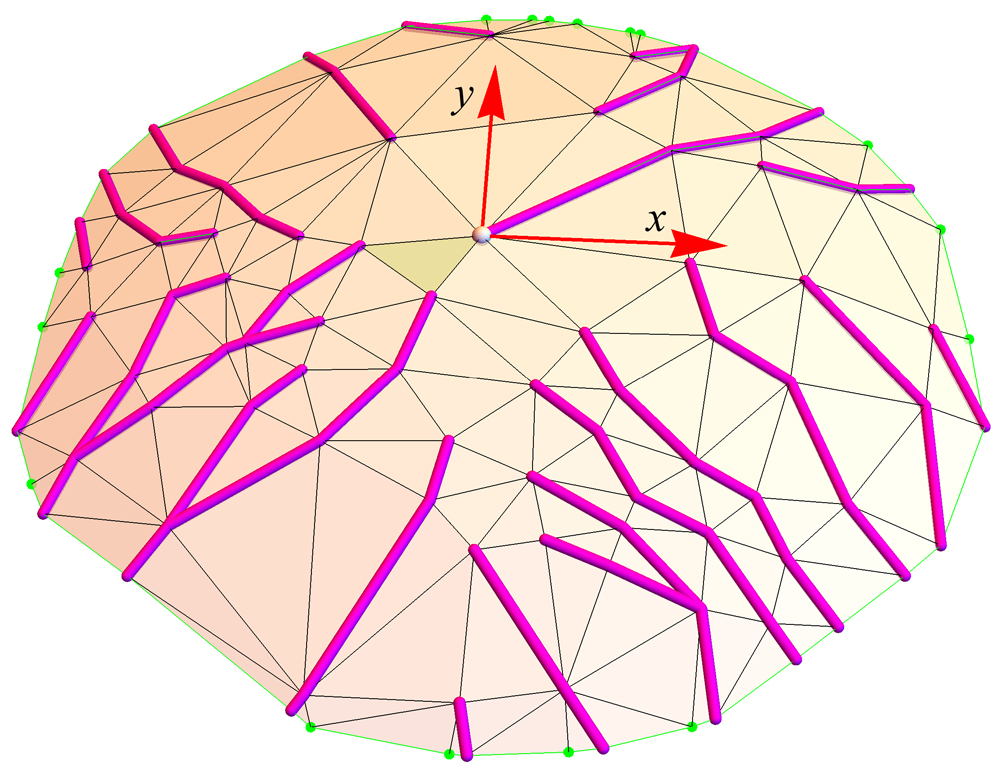}
\caption{The cut forest $\F$ resulting from lifting $\F_\bot$ to the convex cap.
(The marked face is the root of the dual unfolding tree.)}
\figlab{AM-20-30-s1-v100-3D}
\end{figure}
\begin{figure}[htbp]
\centering
\includegraphics[width=0.80\linewidth]{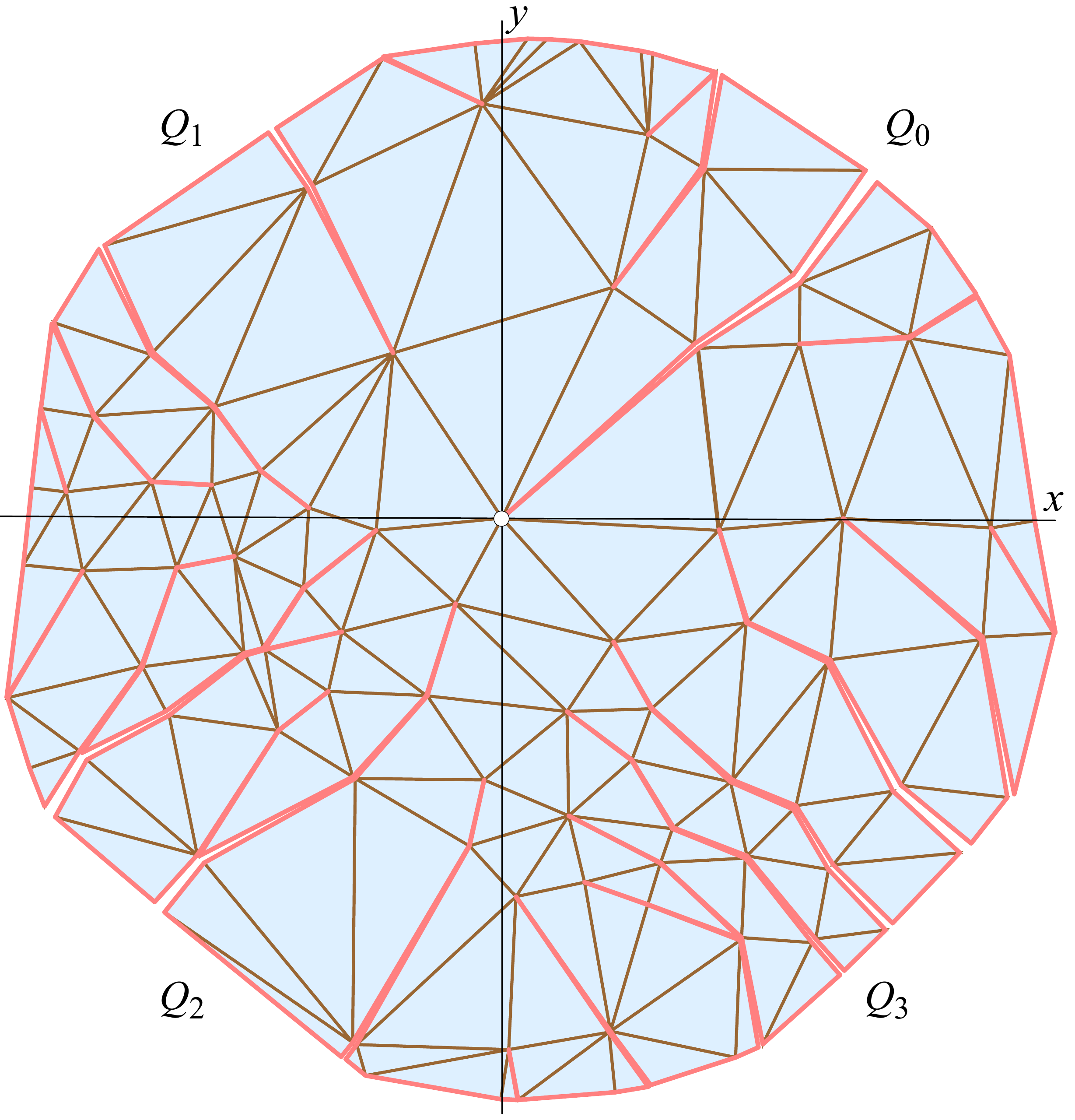}
\caption{The edge-unfolded convex cap. The origin and quadrants
used in Algorithm~1 are indicated.} 
\figlab{AM-20-30-s1-v100-Lay}
\end{figure}
\rednote{I may replace this figure eventually, to remove the root-face highlighting,
which is only distracting here.}

\paragraph{Acknowledgements.}
We thank Debajyoti Mondal for observing that our proof of Theorem~\thmref{PairPath} works for widths other than $90^\circ$.
\rednote{Rev1:
``Add more information to [BBC+16] and cite the CIAC 2017 version of [MS16]."
JOR: Done. Please check your paper.
That Rev1 knew of CIAC for [MS16] is a hint that perhaps one of M or S $=$ Rev1.}

\rednote{JOR: I reduced authors to initials to save space.}


\bibliographystyle{alpha}
\bibliography{geom}
\end{document}